\documentclass[11pt,letterpaper]{article}

\pdfoutput=1

\usepackage[utf8]{inputenc}
\usepackage{fullpage,times}
\usepackage[numbers]{natbib}
\usepackage{authblk}

\usepackage{amsmath,amsfonts,amsthm,amssymb,multirow}
\usepackage{cancel} % for \cancel command
\usepackage{mathtools}

\usepackage{lineno}

\usepackage{parskip} %% used to replace indents with space between paragraphs

\newcommand\thetarel{\operatorname{\theta}}
\newcommand\notthetarel{\operatorname{\cancel\theta}}
\newcommand\thetahrel{\operatorname{\hat\theta}}

\newcommand\gammarel{\operatorname{\gamma}}

\newcommand\gammahrel{\operatorname{\hat\gamma}}
\newcommand\notgammahrel{\operatorname{\cancel{\hat\gamma}}}

\newcommand\coorddiff[3]{{D_{#1}({#2},{#3})}}

\newcommand\edgediff[5]{\left[d_{#1}({#2}, {#4}) - d_{#1}({#2},{#5})\right] - \left[d_{#1}({#3},{#4}) - d_{#1}({#3},{#5})\right]}

\renewcommand\tilde{\widetilde}

\newcounter{names}
\counterwithin{names}{section}

\newtheorem{thm}[names]{Theorem}
\newtheorem{cor}[names]{Corollary}
\newtheorem{lem}[names]{Lemma}
\newtheorem{definition}[names]{Definition}

\title{Isometric Hamming embeddings of weighted graphs}

\author[1]{Joseph Berleant\thanks{To whom correspondence should be addressed: jberleant@gmail.com}}
\author[2]{Kristin Sheridan\thanks{Research performed while at the Department of Electrical Engineering and Computer Science, Massachusetts Institute of Technology, Cambridge, MA}}
\author[3]{Anne Condon}
\author[4]{Virginia Vassilevska Williams}
\author[1]{Mark Bathe}

\affil[1]{Department of Biological Engineering, Massachusetts Institute of Technology, Cambridge, MA}
\affil[2]{Department of of Computer Science, University of Texas at Austin}
\affil[3]{Department of Computer Science, University of British Columbia, Vancouver, Canada}
\affil[4]{Computer Science and Artificial Intelligence Laboratory, Massachusetts Institute of Technology, Cambridge, MA}

\begin{document}
\date{}
\maketitle

\begin{abstract}
% \cite{Laurent1994EmbeddingsOG}
% \cite{Shpectorov}
A mapping $\alpha : V(G) \to V(H)$ from the vertex set of one graph $G$ to another graph $H$ is an \emph{isometric embedding} if the shortest path distance between any two vertices in $G$ equals the distance between their images in $H$. Here, we consider isometric embeddings of a weighted graph $G$ into unweighted Hamming graphs, called Hamming embeddings, when $G$ satisfies the property that every edge is a shortest path between its endpoints. Using a Cartesian product decomposition of $G$ called its \emph{pseudofactorization}, we show that every Hamming embedding of $G$ may be partitioned into Hamming embeddings for each irreducible pseudofactor graph of $G$, which we call its \emph{canonical partition}. This implies that $G$ permits a Hamming embedding if and only if each of its irreducible pseudofactors is Hamming embeddable. This result extends prior work on unweighted graphs that showed that an unweighted graph permits a Hamming embedding if and only if each irreducible pseudofactor is a complete graph. When a graph $G$ has nontrivial pseudofactors, determining whether $G$ has a Hamming embedding can be simplified to checking embeddability of two or more smaller graphs.
\end{abstract}
%% \begin{keyword}



%% \end{keyword}

% \end{frontmatter}

%\linenumbers

\section{Introduction}

\emph{Isometric embeddings}, or distance-preserving mappings from the vertices of one graph to another, are well studied for unweighted graphs but remain relatively unstudied for weighted graphs. Such embeddings have important applications in molecular engineering to design sets of DNA strands with pre-specified binding strengths that generate nontrivial emergent behaviors \cite{eigen1981transfer, short2012promiscuous, zhu2010noncognate, antebi2017combinatorial, malinauskas2014extracellular}, with utility in DNA memory \cite{baum1995building, neel2006semantic, bee2021molecular} and DNA reaction cascades \cite{qian2011scaling, zhang2011dynamic, qian2011neural, cherry2018scaling}. More generally, isometric embeddings are useful whenever a graph's distance metric is of primary interest, and representation in another graph can simplify analysis or manipulation of graph distances. Of particular interest are embeddings into Hamming graphs, or products of complete graphs. For example, this task appears in communication theory because Hamming graphs permit maximally efficient information routing without inspecting the global network structure \cite{graham1971addressing}; in linguistics to relate the closeness of linguistic objects to simpler predicate vector models \cite{firsov1965isometric}; and in coding theory to optimize error-checking codes based on Hamming distance \cite{kautz1958unit}.

Finding isometric embeddings into a particular destination graph or determining their existence is nontrivial, even with simple graphs like Hamming graphs. A large body of work addresses isometric embeddings of unweighted graphs \cite{djokovic1973distance,winkler1984isometric,wilkeit1990isometric,graham1985isometric,feder} but studies of weighted graphs have considered only embeddings of limited classes of graphs into hypercubes \cite{ Shpectorov, DezaLaurentL1,LaurentFewDistances}. Our interest in the isometric graph embedding problem initially stemmed from molecular engineering, in which the metrics we wish to embed are generally complex and rarely fall into a previously studied graph type (Figure \ref{fig:weighted-graphs}a). Na\"ive attempts to convert weighted graphs into unweighted graphs via edge subdivision can produce unweighted graphs without an isometric embedding when the original weighted graph did permit such an embedding (Figure \ref{fig:weighted-graphs}b). In addition, weighted graphs may have multiple non-equivalent isometric embeddings into a single destination graph (Figure \ref{fig:weighted-graphs}c), which does not occur with unweighted graphs \cite{graham1985isometric}. Indeed, the concept of pseudofactorization, which is central to the study of isometric embeddings for unweighted graphs, has only recently been extended to weighted graphs \cite{sheridan2021factorization}.

\begin{figure}
    \centering
    \includegraphics[width=\linewidth]{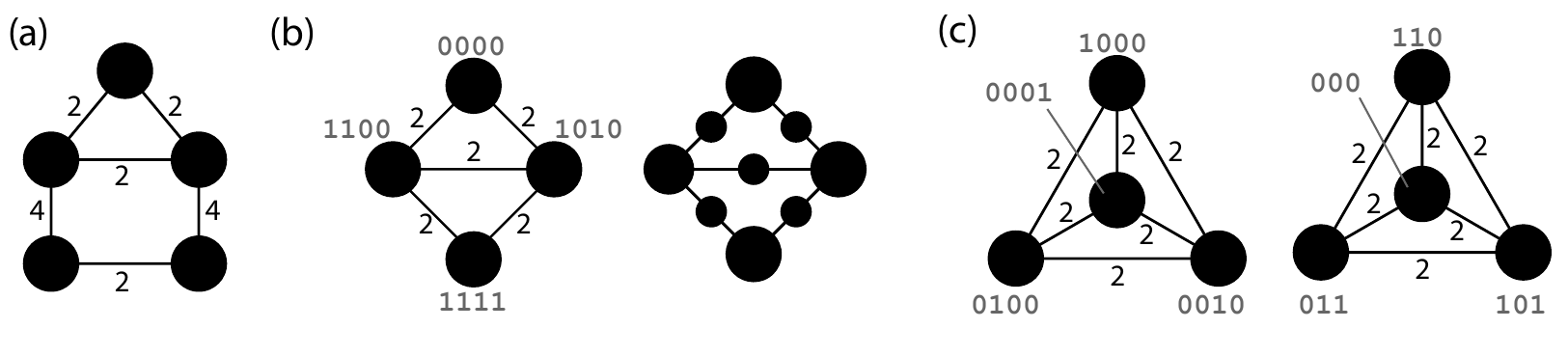}
    \caption{Embedding weighted graphs is more difficult than unweighted graphs, and some guarantees for unweighted graphs no longer hold. a) A simple graph which is embeddable in a hypercube, but which is not covered by previous theorems on hypercube embeddings of weighted graphs. This is hypercube embeddable because it is isometrically embeddable into $(K_2)^4 \times (K_3)^2$. b) A weighted graph that permits a hypercube embedding, but for which the unweighted graph generated via edge subdivision is not. c) A weighted graph, $K_4$ with uniform edge weights of 2, for which multiple non-equivalent hypercube embeddings exist.}
    \label{fig:weighted-graphs}
\end{figure}

The results of this paper apply specifically to those weighted graphs for which each edge is a shortest path between its endpoints, which we call \emph{minimal graphs}. Such graphs are natural to study in the context of isometric embeddings because an edge with weight greater than the distance between its endpoints will not affect the graph's shortest path metric. When only the distances in a graph are of interest, our results also apply to any graph with positive edge weights, because the graph can be made minimal without affecting distances through edge removal. Similarly, our results can apply to any finite metric space by representing it in a minimal graph. Notably, multiple graphical representations generally exist for a given finite metric space, and each may have different embedding properties \cite{imrich1984optimal}.

\subsection{Prior work}

\subsubsection{Pseudofactorization of graphs}

The \emph{Cartesian graph product} of $k\geq 1$ graphs is a graph whose vertex set is the Cartesian product of the vertex sets of each factor graph, and whose edge set is such that each edge of the product graph corresponds to a single edge from a single one of the factors (Figure \ref{fig:he-structure}a).

Cartesian graph products have the important property that every path in the product graph can be decomposed into paths in the factor graphs. This property makes Cartesian products relevant to isometric embeddings, because the the product graph's shortest path distance metric is represented in the distance metrics of the factor graphs. Following the convention of Sheridan et al. \cite{sheridan2021factorization}, a \emph{pseudofactorization} of $G$ is any collection of graphs such that $G$ is isomorphic to an isometric subgraph of their Cartesian product\footnote{The term pseudofactorization is used to distinguish from a \emph{factorization} of a graph $G$, for which $G$ is isomorphic to the full Cartesian product of its factors.}. An \emph{irreducible} graph is one whose pseudofactorizations always include itself. In their work on unweighted graphs, Graham and Winkler \cite{graham1985isometric} found that every connected, unweighted graph has a unique \emph{canonical pseudofactorization} into irreducible graphs, called its canonical pseudofactors (Figure \ref{fig:he-structure}b).

Sheridan et al. \cite{sheridan2021factorization} generalized the notion of pseudofactorization to minimal weighted graphs. In close analogy to unweighted graphs, they showed that any connected minimal graph has a unique canonical pseudofactorization into irreducible minimal graphs, and that the canonical pseudofactorization can be found in polynomial time. To do so, they made use of the Djokovi\'c-Winkler relation $\theta$ on the edge set of a graph, and its transitive closure $\hat\theta$ (Figure \ref{fig:he-structure}c). In particular, they showed that, given any minimal graph, there is a bijection between its equivalence classes under $\hat\theta$ and its canonical pseudofactors.

\subsubsection{Isometric embeddings of graphs into hypercubes and Hamming graphs}

A consequence of Graham and Winkler's work \cite{graham1985isometric} on the canonical pseudofactorization of unweighted graphs is that an unweighted graph permits an isometric embedding into a Hamming graph if and only if every canonical pseudofactor is a complete graph. 

For weighted graphs, work on finding isometric embeddings has been confined to isometric embeddings into hypercubes. In general, determining if an isometric embedding into a hypercube exists for an arbitrary weighted graph is NP-hard \cite{Chvatal}. Efficient (i.e. polynomial-time) methods for determining hypercube embeddability are known for limited classes of weighted graphs. 

Shpectorov \cite{Shpectorov} characterized all graphs with uniform edge weights that could be isometrically embedded into a hypercube, showing that each must be an isometric subgraph of a Cartesian product of complete graphs, cocktail party graphs, and half-cube graphs. Such graphs may be recognized in polynomial time \cite{imrichLinearFactoring,ImrichHammingRecognize,imrichShorterHalved}. Additional weighted graphs for which the existence of isometric embeddings into hypercubes may be determined in polynomial time include line graphs and cycle graphs \cite{DezaLaurentL1}, as well as graphs whose distances are all in the set $\{x, y, x+y\}$ for integers $x,y$ at least one of which is odd \cite{LaurentFewDistances}.

With weighted graphs the isometric embedding of a graph may no longer be unique, as noted previously (e.g., Figure \ref{fig:weighted-graphs}c). A graph with a unique isometric embedding into a hypercube is said to be $l_1$-rigid, and Deza and Laurent \cite{DezaLaurentL1} showed that any uniformly weighted $l_1$-rigid graph must be an isometric subgraph of a uniformly weighted half-cube. Deza and Laurent also considered counting the number of non-equivalent isometric embeddings for uniformly weighted complete graphs, and showed that $K_4$ with uniform edge weights of integer $2k$ always has $k+1$ unique isometric embeddings into hypercubes \cite{laurent1993variety}.

\subsection{Our results}

In this work, we develop a formal relationship between pseudofactorization and isometric embeddings of minimal weighted graphs into unweighted Hamming graphs, called \emph{Hamming embeddings}. All of our results may also be extended to isometric embeddings into hypercubes, called \emph{hypercube embeddings}. One of the two main results of this paper is the following theorem, which states that a minimal graph $G$ permits a Hamming embedding if and only if each of its canonical pseudofactors permits a Hamming embedding. This result may be contrasted with the unweighted case \cite{graham1985isometric}, in which each canonical pseudofactor must \emph{be} a complete graph.

\begin{thm}\label{thm:intro:hamming-embedding}
A minimal weighted graph $G$ has a Hamming embedding if and only if a Hamming embedding exists for each of its canonical pseudofactors. Similarly, $G$ has a hypercube embedding if and only if each of its canonical pseudofactors does.
\end{thm}

Theorem \ref{thm:intro:hamming-embedding} is proven by Theorem \ref{thm:hamming-iff} and Corollary \ref{cor:hypercube-iff} in Section \ref{sec:wg-structure}. We briefly sketch the proof and the novel contributions necessary for it here. A Hamming embedding of weighted graph $G = (V(G),E(G),w_G)$ can be written as a mapping $\eta : V(G) \to \Sigma^m$, where $m$ is the embedding \emph{dimension}, $\Sigma$ is the embedding \emph{alphabet}, and the distance between two elements of $\Sigma^m$ is given by Hamming distance. $\eta$ may be partitioned, that is, the columns may be grouped to form embeddings of lower dimension (though these are not generally isometric). The second main result of this manuscript is a proof of existence for a special \emph{canonical partition} of $\eta$ that provides isometric Hamming embeddings for each canonical pseudofactor of $G$. 

\begin{thm}\label{thm:intro:canonical-partition}
Let $G = (V(G),E(G),w_G)$ be a minimal weighted graph and $\eta$ a Hamming (hypercube) embedding of $G$. Let the canonical pseudofactors of $G$ be the graphs $G_1, \dots, G_n$. Then there exists a partition of $\eta$ into embeddings $\eta^1, \dots, \eta^n$ such that $\eta^i$ forms a Hamming (hypercube) embedding of $G_i$.
\end{thm}

Theorem \ref{thm:intro:canonical-partition} is proven by Theorems \ref{thm:bijection} and \ref{thm:structure}. To prove Theorem \ref{thm:intro:hamming-embedding}, we use Theorem \ref{thm:intro:canonical-partition} to show that if $G$ is Hamming embeddable then its pseudofactors are also Hamming embeddable. The converse is easily shown to be true, by concatenating Hamming embeddings of the canonical pseudofactors. To construct the canonical partition of $\eta$, we introduce a novel relation on the coordinates of $\eta = (\eta_1, \ldots, \eta_m)$, called $\gamma$. Informally, two coordinates are related by $\gamma$ if the corresponding digits of $\eta$ change across some edge in $E$. The transitive closure of $\gamma$ is an equivalence relation $\gammahrel$, and its equivalence classes define a partition of the digits of $\eta$, which is the canonical partition of $\eta$ (Figure \ref{fig:he-structure}c-d). Much of our effort is spent proving the existence of a bijection between the equivalence classes under $\thetahrel$ (i.e. sets of edges) and the equivalence classes under $\gammahrel$ (i.e. sets of coordinates), which we use to construct a bijection between the canonical pseudofactorization of $G$ and the canonical partition of any Hamming embedding of $G$. As a final step, we prove that the embeddings of this partition form Hamming embeddings of each canonical pseudofactor (Figure \ref{fig:he-structure}e).

\begin{figure}
    \centering
    \includegraphics{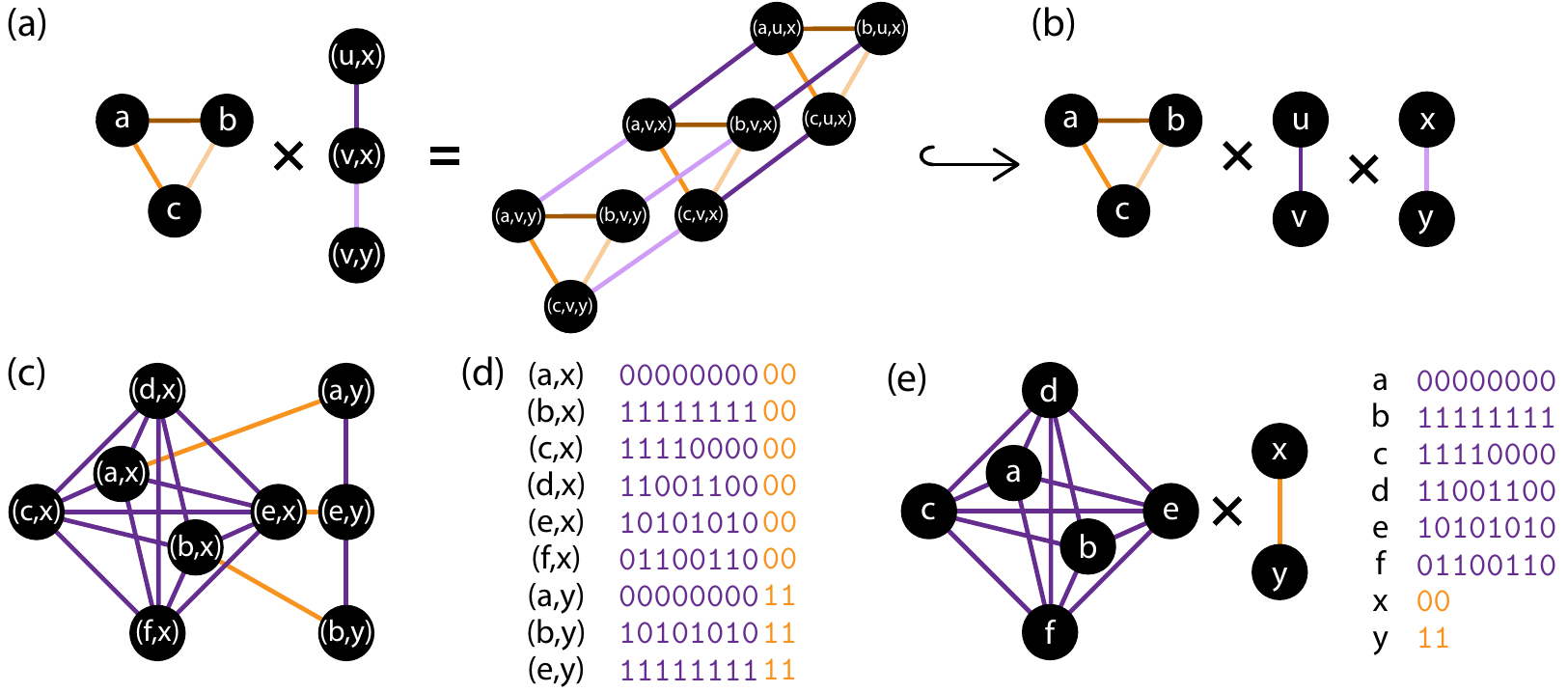}
    \caption{(a) Illustration of the Cartesian graph product. Edge colors indicate correspondence between edges in the factor and product graphs. If edges are weighted, two edges of the same color will have the same weight. (b) This graph product is isometrically embeddable in a product of three irreducible graphs ($K_3$, $K_2$, and $K_2$), which form its canonical pseudofactorization. (c) A weighted graph, with purple edges of weight 4 and orange edges of weight 2. Edge colors also indicate equivalence classes under the transitive closure $\thetahrel$ of the Djokovi\'c-Winkler relation. (d) A hypercube embedding of this graph, with digits colored to indicate equivalence classes under $\gammahrel$. $\gammahrel$ is the transitive closure of $\gamma$, which relates coordinates whose digits change together across some edge. (e) Our results show that the hypercube embedding in (d) may be partitioned into a hypercube embedding for each  of the two canonical pseudofactors. The same is true of any Hamming embedding of a minimal graph.}
    \label{fig:he-structure}
\end{figure}

Using Theorem \ref{thm:intro:canonical-partition}, we are also able to prove the following result on the number of non-equivalent Hamming embeddings or hypercube embeddings of $G$. 
We define equivalence of two Hamming embeddings in the same way as Winkler in \cite{winkler1984isometric}; the formal definition is given in Section \ref{sec:preliminaries}. Informally, two Hamming embeddings are equivalent if they can be made identical by permuting coordinates or coordinate values (i.e., swapping some number of $\eta_i$ and $\eta_j$ and/or swapping the values of a particular $\eta_i$). The following theorem is proven by Corollary \ref{cor:number-embeddings}.

\begin{thm}
Given a minimal weighted graph $G$, the number of non-equivalent Hamming embeddings of $G$ is the product of the number of non-equivalent Hamming embeddings of each of its pseudofactors. Similarly, the number of non-equivalent hypercube embeddings of $G$ is the product of the number of non-equivalent hypercube embeddings of each of its pseudofactors.
\end{thm}

These theorems imply an important structure of any Hamming embedding of a graph $G$: such an embedding must be equivalent to a concatenation of Hamming embeddings for each canonical pseudofactor of $G$. As a consequence, the existence of a Hamming embedding of $G$ implies one for each pseudofactor. The converse is easily shown to be true also. In practice, these results mean that we may recognize graphs that do not permit a Hamming embedding by analyzing the pseudofactors, which may be significantly smaller. They also allow us to extend polynomial-time results for determining Hamming or hypercube embeddability of a graph to graphs whose pseudofactors can be determined as Hamming or hypercube embeddable in polynomial time. These include Cartesian products of line graphs, cycle graphs, graphs with distances in $\{x,y,x+y\}$ for integers $x,y$ at least one odd, and graphs with uniform edge weights \cite{DezaLaurentL1, LaurentFewDistances, Shpectorov}, along with any other graphs later found to be Hamming or hypercube embeddable. Significant additional work remains in characterizing classes of irreducible weighted graphs for which Hamming or hypercube embeddings may be constructed efficiently. However, our work significantly eases the problem of finding Hamming or hypercube embeddings for graphs with nontrivial pseudofactorizations, and is also a step forward in better understanding isometric embeddings of weighted graphs into more complex destination graphs.

\section{Preliminaries \label{sec:preliminaries}}

In this paper, all graphs are finite, connected, and undirected, and we use $V(G)$, $E(G)$, and $w_G : E(G) \to \mathbb{Z}_{>0}$ to denote the vertex set, edge set, and weight function of a graph $G$, respectively. $G$ being unweighted is equivalent to $w_G(e) = 1$ for all $e \in E(G)$. An edge between vertices $u,v \in V(G)$ is written $uv$ or $vu$; since all edges are undirected, $uv \in E(G) \iff vu \in E(G)$. The distance from $u$ to $v$, $u,v \in V(G)$, is the minimum edge weight sum along a path from $u$ to $v$, and is given by the shortest path metric denoted $d_G : V(G) \times V(G) \to \mathbb{Z}_{\ge 0}$.

The following definition is due to Sheridan et al. \cite{sheridan2021factorization}:
\begin{definition}\label{def:minimal-graph}
A graph $G$ is a \textbf{minimal graph} if  and only if every edge in $E(G)$ forms a shortest path between its endpoints. That is, $w_G(uv) = d_G(u,v)$ for all $uv \in E(G)$.
\end{definition}

In this manuscript, we assume all graphs are minimal. Note that any unweighted graph is minimal and that any non-minimal graph may be made minimal by removing any edges not satisfying the condition in Definition \ref{def:minimal-graph}.
 
The Cartesian graph product of one or more graphs $G_1, \dots, G_n$ is written $G = G_1 \times \cdots \times G_n$ or $G = \prod_{i=1}^n G_i$. $G$ is defined as $V(G) = V(G_1) \times \cdots \times V(G_n)$,  with two vertices $(u_1,\dots,u_n)$ and $(v_1,\dots,v_n)$ adjacent if and only if there is exactly one $\ell$ such that $u_\ell v_\ell \in E(G_\ell)$ and $u_i=v_i$ for all $i\ne \ell$, and $w_G(uv) = w_{G_\ell}(u_\ell v_\ell)$ for $\ell$ chosen as above (Figure \ref{fig:he-structure}a). The Cartesian graph product has the following important distance property:
\begin{equation}\label{eqn:cartesian-distance-property}
d_G(u,v) = \sum_{i=1}^n d_{G_i}(u_i, v_i)
\end{equation}
where $u = (u_1,\dots, u_n)$ and $v = (v_1, \dots, v_n)$.
This implies that any path in $G$ can be decomposed into a set of paths in the $G_i$, with the path length in $G$ equal to the sum of the path lengths in the $G_i$.

A graph embedding $\pi : V(G) \to V(G^*)$ of a graph $G$ into a graph $G^*$ maps vertices of $G$ to those of $G^*$. If $\pi$ satisfies $d_G(u,v) = d_{G^*}(\pi(u),\pi(v))$ for all $u,v \in V(G)$, then $\pi$ is an isometric embedding. When such a $\pi$ exists, we say that $G \hookrightarrow G^*$. For convenience, we let $d_\pi(u,v) = d_{G^*}(\pi(u),\pi(v))$. 

When $G$ is isomorphic to an isometric subgraph of a Cartesian graph product (e.g., Figure \ref{fig:he-structure}b), we call the set of multiplicands a pseudofactorization. The following definition is again due to Sheridan et al. \cite{sheridan2021factorization}:
\begin{definition}\label{def:pseudofactorization}
Consider graphs $G$ and $G^* = \prod_{i=1}^n G^*_i$. If an embedding $\pi : V(G) \to V(G^*)$, $\pi = (\pi_1, \dots, \pi_n)$, exists satisfying the following criteria:
\begin{enumerate}
\item $d_G(u,u') = d_{G^*}(\pi(u), \pi(u'))$,
\item $uv \in E(G)$ implies $\pi(u)\pi(v) \in E(G^*)$ and $w_G(uv) = w_{G^*}(\pi(u)\pi(v))$,
\item every vertex in $G^*_i$ is in the image of $\pi_i$, $1 \le i \le n$, and
\item every edge $u_iv_i$ in $G^*_i$ equals $\pi_i(u)\pi_i(v)$ for some $uv \in E(G)$
\end{enumerate}
then we say the set $\{ G^*_1, \dots, G^*_n \}$ is a \textbf{pseudofactorization} of $G$ and refer to each $G^*_i$ as a \textbf{pseudofactor}.
\end{definition}

A graph $G$ is \emph{irreducible} if all its pseudofactorizations include itself as a pseudofactor. A pseudofactorization into irreducible graphs is called an irreducible pseudofactorization. For convenience, we assume that a pseudofactorization does not include $K_1$, except when $G = K_1$.

Informally, the definition of pseudofactorization requires both that $G$ be isometrically embeddable into ${G^*}$ \emph{and} that edges be preserved within this embedding. This second condition is a natural one for manipulating graph structures, but may be less applicable to other situations (e.g., finite metric spaces). The final two conditions ensure that there are no unnecessary vertices and edges in the pseudofactors (or any graph would be a pseudofactor of the graph in question, as $G$ is an isometric subgraph of $G\times H$ for any graph $H$).

Sheridan et al. \cite{sheridan2021factorization} showed that an irreducible pseudofactorization of a weighted graph $G$ is unique. This pseudofactorization is called its \emph{canonical pseudofactorization}, and the corresponding isometric embedding implied by the first condition of Definition \ref{def:pseudofactorization} is its \emph{canonical embedding}. The authors used the Djokovi\'c-Winkler relation $\theta$, which we restate here as has been defined elsewhere \cite{graham1985isometric, sheridan2021factorization}. For a graph $G$, two edges in the graph $uv,ab\in E(G)$ are related by $\theta$ if and only if: 
\begin{align}
    [d_G(u,a)-d_G(u,b)]-[d_G(v,a)-d_G(v,b)] \neq 0. \label{eqn:theta-diff}
\end{align}
We note that $\theta$ is symmetric and reflexive. Let the equivalence relation $\hat\theta$ be the transitive closure of $\theta$.

A constructive algorithm \cite{sheridan2021factorization} proves the existence of the canonical pseudofactorization. This algorithm creates exactly one canonical pseudofactor from each element of $E(G)/\thetahrel$, the set of equivalence classes under $\thetahrel$. We make the following note about the canonical pseudofactors produced by this algorithm:

\begin{thm}
\label{thm:prelim:pseudofactor-path}
Consider the canonical pseudofactorization of $G$, $\{G_1, \dots, G_n\}$, as constructed by Algorithm 1 of \cite{sheridan2021factorization}, and assume each $G_i$ was constructed from equivalence class $E_i \in E(G)/\hat\theta$. Let $\pi = (\pi_1, \dots, \pi_n)$ be the canonical embedding of $G$ into $\prod_{i=1}^n G_i$. For $u,v \in V(G)$ and a path $P$ from $u$ to $v$, let $P_j$ be the subsequence of edges in $P$ that are in $E_j$. Then there is a path in $G_i$ of length equal to the sum of the edge weights in $P_j$.
\end{thm}

\begin{proof}
By inspection of Algorithm 1, two vertices connected by an edge not in $E_j$ will be identified in $G_j$.
Since $\pi$ is an isometric embedding of $G$, $w_G(uv) = d_G(u,v) = \sum_{i=1}^n d_{\pi_i}(u,v) = d_{\pi_j}(u,v)$, where the final equality is because, by the previous sentence, all summands are zero except when $i=j$.
As a result, for each $uv \in P_j$ there is a path from $\pi_j(u)$ to $\pi_j(v)$ of length $w_G(uv)$. Any edges in $P$ but not $P_j$ do not change the value of $\pi_j$. So adjacent edges in $P_j$ share an endpoint, and there is a path of length equal to the sum of the edge weights in $P_j$.
\end{proof}

Finally, we introduce the following notation for our discussions of Hamming graphs. If $H$ is a Hamming graph $H$, then $V(H) = \Sigma^m$, where $\Sigma$ is the alphabet and $m$ the dimension of $H$, and distance in $H$ is given by the Hamming distance between pairs of vertices. $H$ is a hypercube graph if $|\Sigma| = 2$. For Hamming or hypercube graph $H$, an isometric embedding $\eta : V(G) \to V(H)$, is a Hamming or hypercube embedding, respectively. We use subscripts to refer to individual letters in the image of $\eta = (\eta_1, \dots, \eta_m)$ or, equivalently, $\eta = \eta_1\cdots\eta_m$. The indices $[m]$ of $\eta$ are its \emph{coordinates}, where $[m]$ is the first $m$ positive integers, and each $\eta_i$, $1\le i \le m$, is a \emph{digit}. Two for more embeddings $\eta^1, \dots, \eta^n$ of dimensions $m_1, \dots, m_n$ may be concatenated, forming an embedding $\eta = \eta^1 \cdots \eta^n = \eta^1_1 \cdots \eta^1_{m_1} \cdots \eta^n_1 \cdots \eta^n_{m_n}$.

Two Hamming embeddings $\eta$, $\eta'$ are equivalent if there exists a permutation $\sigma$ of the $m$ coordinates and $m$ functions $\beta_i$, $1\le i\le m$, such that $\eta'(u) = \beta_1(\eta_{\sigma(1)}(u))\cdots \beta_m(\eta_{\sigma(m)}(u))$ for all $u \in V(G)$. A Hamming embedding of $G$ is unique if it is equivalent to all other Hamming embeddings of $G$. We assume there are no unnecessary digits in $\eta$ (i.e., every digit changes across some edge), which is analogous to our assumption that $K_1$ is not used in any pseudofactorization.

For any two vertices $u,v \in V(G)$, we define the function $D_\eta$ as follows:
\begin{equation}\label{eqn:coorddiff}
\coorddiff{\eta}{u}{v} = \{ j \in [m] : \eta_j(u) \ne \eta_j(v) \},
\end{equation}
noting that a graph $G$ with Hamming embedding $\eta : V(G) \to V(H)$ will have $d_G(u,v) = d_\eta(u,v) = |D_\eta(u,v)|$. When $uv$ is an edge, we say that the digits indicated by $D_\eta(u,v)$ change across $uv$. We also introduce a relation $\gamma$, which relates two coordinates of an embedding if both corresponding digits change across any edge:
\begin{equation}\label{eqn:gamma}
j \gammarel j' \iff \exists uv \in E(G), \{j,j'\} \subseteq \coorddiff{\eta}{u}{v}.
\end{equation}
Let $\gammahrel$ be the transitive closure of $\gamma$. Since $\gamma$ is symmetric and reflexive, $\gammahrel$ is an equivalence relation. 

A \emph{partition} $\{\eta^1, \dots, \eta^{n}\}$ of an $m$-dimensional Hamming embedding $\eta$ is defined by a partition of its coordinates $[m]$, $\{J_1, \dots, J_{n}\}$, with each $\eta^i$ equal to the projection of $\eta$ onto the coordinates in $J_i$. Let $[m]/\hat\gamma$ be the equivalence classes of $[m]$ under $\hat\gamma$. Then $[m]/\hat\gamma$ defines a partition of $\eta$, which we call its \emph{canonical partition}. This terminology is motivated by Theorem \ref{thm:bijection}, which guarantees a bijection between the canonical partition of $\eta$ and the canonical pseudofactorization of $G$.

\section{Structure of Hamming embeddings of weighted graphs \label{sec:wg-structure}}

In this section, we begin by constructing a bijection between the canonical partition of any Hamming embedding of $G$ and its set of canonical pseudofactors. This result is used to prove the two main findings of the paper: Theorem \ref{thm:structure}, which proves that the canonical partition of a Hamming embedding of $G$ forms a Hamming embedding for each canonical pseudofactor; and Theorem \ref{thm:hamming-iff}, which proves that $G$ permits a Hamming embedding if and only if each of its pseudofactors also permits a Hamming embedding.

\begin{lem}\label{lem:structure:1}
Let $G$ be a weighted graph and $\eta$ a Hamming embedding of $G$. Let $uv, u'v' \in E(G)$, $j \in \coorddiff{\eta}{u}{v}$, and $j' \in \coorddiff{\eta}{u'}{v'}$. Then we have
\[ uv \thetahrel u'v' \iff j \gammahrel j'. \] 
\end{lem}

\begin{proof}
To see that $uv \thetahrel u'v' \Rightarrow j \gammahrel j'$, observe that $j \notgammahrel j'$ implies that $\coorddiff{\eta}{u}{v}$ and $\coorddiff{\eta}{u'}{v'}$ are disjoint. Thus,
\begin{multline*}
\edgediff{G}{u}{v}{u'}{v'} \\
= \sum_{s \in [m]} \edgediff{\eta_s}{u}{v}{u'}{v'} = 0
\end{multline*}
because each summand is nonzero only if both $\eta_s(u) \ne \eta_s(v)$ and $\eta_s(u') \ne \eta_s(v')$. So $j \notgammahrel j' \Rightarrow uv \notthetarel u'v'$, and $uv \thetarel u'v' \Rightarrow j \gammahrel j'$. If instead $uv \thetahrel u'v'$ then some sequence of edges $(e_i)_{i = 1}^l$ satisifes $uv \thetarel e_1 \thetarel \cdots \thetarel e_{l} \thetarel u'v'$, and the transitivity of $\gammahrel$ implies $j \gammahrel j'$.

To prove $j \gammahrel j' \Rightarrow uv \thetahrel u'v'$, let $S = \coorddiff{\eta}{u}{v} \cap \coorddiff{\eta}{u'}{v'}$, the set of all coordinates that change across both $uv$ and $u'v'$. We prove the statement in two cases:

\begin{description}
\item[Case 1 ($|S| > 0$):] Consider any path $P$ with $l$ edges, beginning with $uv$ and ending with $u'v'$, $l \ge 2$. We show by induction on $l$ that $|S| > 0 \Rightarrow uv \thetahrel u'v'$. As a base case, consider $l=2$ and without loss of generality let $P = (u, v=u', v')$. As above, we consider
\begin{multline*}
\edgediff{\eta}{u}{v}{u'}{v'} \\= \sum_{s \in S} \edgediff{\eta_s}{u}{v}{u'}{v'}
\end{multline*}
with the sum restricted to coordinates $s \in S$ for which the summand is nonzero. Because $v=u'$, $\eta_s(v) = \eta_s(u')$, so $d_{\eta_s}(v,u') = 0$ and $d_{\eta_s}(u,u') = d_{\eta_s}(v,v') = 1$. Thus, each summand is at least +1 and the summation is at least $+|S|$ and $|S| > 0 \Rightarrow uv \thetarel u'v'$. Now, consider the case $l > 2$ and $P = (u_0 = u, u_1 = v, \dots, u_{l-1}=u', u_{l}=v')$. If every edge $u_{i-1}u_i$, $1 < i < l$, has $S$ and $ \coorddiff{\eta}{u_{i-1}}{u_i}$ disjoint, then $\eta_s(v) = \eta_s(u')$ for all $s\in S$ and as in the base case we have $|S| > 0 \Rightarrow uv \thetarel u'v'$. If some edge $u_{i-1}u_i$ has $S$ and $\coorddiff{\eta}{u_{i-1}}{u_i}$ not disjoint, then consider the subpaths $P_1 = (u_0, u_1, \dots, u_{i})$ and $P_2 = (u_{i-1}, u_i, \dots, u_l)$, with $l_1$ and $l_2$ edges, respectively. Clearly $2 \le l_1, l_2 < l$, so by induction we conclude that $uv \thetahrel u_{i-1}u_i \thetahrel u'v'$.
\item[Case 2 ($|S| = 0$):] We are given $j \gammahrel j'$, so construct a sequence of coordinates $j_0 = j, \dots, j_l = j'$ such that $j_{i-1} \gammarel j_i$ for all $1 \le i \le l$. For each pair of coordinates $j_{i-1}$, $j_i$ there is some edge $e_i$ across which both $\eta_{j_{i-1}}$ and $\eta_{j_i}$ change. By Case 1, we have $uv \thetahrel e_1$, $e_l \thetahrel u'v'$, and $e_{i-1} \thetahrel e_i$ for $1 < i \le l$. Thus, $uv \thetahrel u'v'$ as desired. 
% \qedhere
\end{description}
\end{proof}

\begin{thm}
\label{thm:bijection}
Let $G$ be a weighted graph and $\eta$ an $m$-dimensional Hamming embedding of $G$. Then there is a bijection between $E(G) / \hat\theta$ and $[m] / \hat\gamma$. It follows that there is a natural bijection from the set of canonical pseudofactors of $G$ to the canonical partition of $\eta$.
\end{thm}

\begin{proof}
We first construct a bijection from $E(G)/\hat\theta$ to $[m]/\hat\gamma$. Let $[uv]_{\hat\theta} \in E(G)/\hat\theta$ contain $uv \in E(G)$ and $[j]_{\hat\gamma} \in [m]/\hat\gamma$ contain $j \in [m]$. Then we consider the mapping $f_2 : [uv]_{\hat\theta} \mapsto [j]_{\hat\theta}$ where $j \in \coorddiff{\eta}{u}{v}$. By Lemma \ref{lem:structure:1}, $f_2$ is injective because for  $j \in \coorddiff{\eta}{u}{v}$ and $j' \in \coorddiff{\eta}{u'}{v'}$, $f_2([uv]_{\hat\theta}) = f_2([u'v']_{\hat\theta})$ implies $[j]_{\hat\gamma} = [j']_{\hat\gamma}$, so $j \gammahrel j'$ and thus $uv \thetahrel u'v'$. $f_2$ is surjective because each digit of $\eta$ changes over some edge. This proves the first assertion of the theorem. For the second assertion, we use the bijection between $E(G)/\hat\theta$ and the canonical pseudofactorization (see Section \ref{sec:preliminaries}) that takes each $G_i$ to the $E_i \in E(G)/\hat{\theta}$ that was used to construct it, where $\{G_1,\dots,G_n\}$ is the canonical pseudofactorization. We also have a bijection $f_3$ from $[m]/\hat\gamma$ to the canonical partition of $\eta$, since the elements of $[m]/\hat\gamma$ were used to form that partition. Thus, the mapping $f = f_3 \circ f_2 \circ f_1$ is a bijection from the canonical pseudofactors of $G$ to the canonical partition of $\eta$.
\end{proof}

\begin{thm}
\label{thm:structure}
Let weighted graph $G$ have canonical pseudofactors $G_1, \dots, G_{n}$, with $\pi : V(G) \to V(\prod_{i=1}^n)$ the canonical embedding of $G$. Let $\eta$ be an $m$-dimensional Hamming embedding of $G$ with canonical partition $\eta^1, \dots, \eta^{n}$. Assume without loss of generality that the natural bijection of Theorem \ref{thm:bijection} maps $G_i$ to $\eta^i$ for each $i, 1 \le i \leq  n$. Then there is an embedding $\tilde{\eta}^i$ implied by $\eta^i = \tilde{\eta}^i \circ \pi_i$, which is a Hamming embedding of $G_i$ for each $i, 1 \le i \leq  n$.
\end{thm}

\begin{proof}
%Fix any two vertices $u,v \in V(G)$.
%We know that the canonical embedding is isometric, so 
%\[ d_G(u,v) = \sum_{i=1}^n d_{G_i}(g_i(u), g_i(v)), \]
%Because $\eta$ is an isometric embedding of $kG$, we also have that
%\[ k\cdot d_G(u,v) = d_{\eta}(u, v) = \sum_{i=1}^n d_{\eta^i}(u, v). \]
For convenience, let $\pi = (\pi_1, \dots, \pi_{n})$, so that the $\pi_i$ are embeddings of $G$ into each canonical pseudofactor.

Fix any two vertices $u,v \in V(G)$ and consider a shortest path $P$ from $u$ to $v$.
For each pseudofactor $G_i$, let $E_i$ be the equivalence class under $\hat\theta$ from which $G_i$ was generated and let $c_i$ be the sum of the edge weights for edges along $P$ that are in $E_i$. By Theorem \ref{thm:prelim:pseudofactor-path}, $P$ implies a path in $G_i$ of length $c_i$, each $d_{\pi_i}(u,v) \le c_i$. Further, each edge along $P$ contributes to exactly one $c_i$, so we have 
\[ d_G(u,v) = \sum_{i=1}^{n} c_i = \sum_{i=1}^{n} d_{\pi_i}(u,v), \]
where the second equality is due to the fact that $\pi$ is an isometric embedding. Thus each $d_{\pi_i}(u,v) = c_i$.
Now take $\eta^i$, for which we know, based on the construction of Theorem \ref{thm:bijection}, that $w_G(e)$ digits change across any edge $e \in E_i$ and no digits change across any other edge. Then $d_{\eta^i}(u,v) \le c_i$. In fact, we have
\[ d_G(u,v) = \sum_{i=1}^n c_i = \sum_{i=1}^{n} d_{\eta^i}(u, v), \]
where again the second equality is because $\eta$ is an isometric embedding. As before, this implies that $d_{\eta^i}(u, v) = c_i$. 
Thus, $d_{\eta^i}(u, v) = d_{\pi_i}(u,v)$ as desired.

We construct a Hamming embedding $\tilde\eta^i$ for $G_i$ that satisfies $\eta^i = \tilde\eta^i \circ \pi_i$. $\pi_i$ is not unambiguously invertible because multiple nodes in $V(G)$ may map to the same $u_i$ in $G_i$. However, we may let $\pi_i^{-1} : V(G_i) \to V(G)$ map $u_i \in V(G_i)$ to any $u \in V(G)$ such that $\pi_i(u) = u_i$. Let $\tilde\eta^i = \eta^i \circ \pi_i^{-1}$. Then for $u_i,v_i \in V(G_i)$ such that $\pi^{-1}_i(u_i) = u$ and $\pi^{-1}_i(v_i) = v$,
$\tilde\eta^i(u_i) = \eta^i(\pi_i^{-1}(u_i)) = \eta^i(u)$, and so $d_{\tilde\eta^i}(u_i,v_i) = d_{\eta^i}(u,v) = d_{\pi_i}(u,v) = d_{G_i}(\pi_i(u), \pi_i(v)) = d_{G_i}(u_i,v_i)$. Thus, $\tilde\eta^i$ is a Hamming embedding of $G_i$.
\end{proof}

\begin{thm}
\label{thm:hamming-iff}
Let $G$ be a minimal graph with canonical pseudofactors $G_1, \dots, G_n$. Then $G$ is Hamming embeddable if and only if each $G_i$ is Hamming embeddable.
\end{thm}
\begin{proof}
Let $\pi : V(G) \to \prod_{i=1}^n V(G_i)$ be the canonical embedding of $G$, with $\pi = (\pi_1, \dots, \pi_n)$. 

If $G$ is Hamming embeddable then by Theorem \ref{thm:structure} we may construct Hamming embeddings of each pseudofactor. 

If every pseudofactor is Hamming embeddable, then let $\tilde\eta^i$ be a Hamming embedding for pseudofactor $G_i$. Let $\eta^i = \tilde\eta^i \circ \pi_i$. If $\pi_i(u) = u_i$ and $\pi_i(v) = v_i$, we have $\eta^i(u) = \tilde\eta^i(\pi_i(u)) = \tilde\eta^i(u_i)$, so $d_{\eta^i}(u,v) = d_{\tilde\eta^i}(u_i,v_i) = d_{G_i}(u_i, v_i)$. So we may concatenate the $\eta^i$ to form $\eta$ such that
\begin{align}
d_{\eta}(u,v) &= \sum_{i=1}^n d_{\eta^i}(u,v) \\
&= \sum_{i=1}^n d_{G_i}(u_i,v_i) \\
&= d_{G}(u,v)
\end{align}
where the final equality is due to the fact that $\pi$ is an isometric embedding of $G$ into $\prod_{i=1}^n G_i$. Thus, $\eta$ is a Hamming embedding of $G$.
\end{proof}

\begin{cor}
\label{cor:hypercube-iff}
A weighted graph $G$ is hypercube embeddable if and only if each of its canonical pseudofactors is hypercube embeddable.
\end{cor}

\begin{proof}
If $G$ has a hypercube embedding $\eta$, then each element of the canonical partition of $\eta$ is also a hypercube embedding. Thus each canonical pseudofactor of $G$ is hypercube embeddable because it has a hypercube embedding formed from an element of the canonical partition of $\eta$. Conversely, if each canonical pseudofactor of $G$ has a hypercube embedding, these may be concatenated to form a hypercube embedding of $G$.
\end{proof}

\begin{cor}\label{cor:number-embeddings}
Given a weighted graph $G$, the number of non-equivalent Hamming embeddings of $G$ is the product of the number of non-equivalent Hamming embeddings of each of its pseudofactors. Similarly, the number of non-equivalent hypercube embeddings of $G$ is the product of the number of non-equivalent hypercube embeddings of each of its pseudofactors.
\end{cor}

\begin{proof}
Let $\pi$ be the canonical embedding of $G$, $\pi : V(G) \to V(\prod_{i=1}^n G_i)$, and $\pi = (\pi_1, \dots, \pi_{n})$.

Take any two non-equivalent Hamming embeddings $\eta$ and $\zeta$ of $G$. 
Each element $\eta^i$ of the canonical partition of $\eta$ corresponds to a Hamming embedding $\tilde\eta^i$ of $G_i$, with $\eta^i = \tilde\eta^i \circ \pi_i$. 
This is similarly true of $\zeta$. Now if, for all $1 \le i \le n$, $\tilde\eta^i$ is equivalent to $\tilde\zeta^i$, then each $\tilde\eta^i$ can be made identical to $\tilde\zeta^i$ by permuting coordinates and coordinate values. 
Because $\eta^i = \tilde\eta^i \circ \pi_i$ and $\zeta^i = \tilde\zeta^i \circ \pi_i$, this implies that $\eta$ can be made identical to $\zeta$ by the same process, so $\eta$ and $\zeta$ are equivalent. 
Thus, non-equivalent $\eta$ and $\zeta$ must have some $\tilde\eta^i$ not equivalent to $\tilde\zeta^i$. 
That is, any non-equivalent $\eta$ and $\zeta$ will correspond to Hamming embeddings of the $G_i$ that are distinct under equivalence, so the product of the number of non-equivalent Hamming embeddings of the $G_i$ is at least the number of non-equivalent Hamming embeddings of $G$.

Now let $\tilde\eta^1, \dots, \tilde\eta^{n}$ and $\tilde\zeta^1, \dots, \tilde\zeta^{n}$ be such that each $\tilde\eta^i$ and $\tilde\zeta^i$ are Hamming embeddings of $G_i$. Let $\eta^i = \tilde\eta^i \circ \pi_i$ and $\zeta^i = \tilde\zeta^i \circ \pi_i$, and consider the concatenations of the $\eta^i$ and the $\zeta^i$, $\eta = \eta^1\cdots\eta^{n}$ and $\zeta = \zeta^1\cdots\zeta^{n}$. Note that $\eta$ and $\zeta$ form Hamming embeddings of $G$. If $\eta$ and $\zeta$ are equivalent, then $\eta$ may be made identical to $\zeta$ by permuting coordinates and coordinate values. Note that any such permutation must map each coordinate in $\eta^i$ to a coordinate in $\zeta^i$, because the corresponding digits necessarily change across the same edges. This allows the permutation from $\eta$ to $\zeta$ to be decomposed into permutations from $\eta^i$ to $\zeta^i$, so each $\tilde\eta^i$ and $\tilde\zeta^i$ must be equivalent. 
From this, we conclude that, for any $1 \le i \leq  n$, if $\tilde\eta^i$ and $\tilde\zeta^i$ are not equivalent then their corresponding $\eta$ and $\zeta$ generated by the above process are also not equivalent. This indicates that the number of non-equivalent Hamming embeddings of $G$ is not less than the product of the number of non-equivalent Hamming embeddings of the $G_i$. This completes the proof.

This can be proven identically for counting hypercube embeddings, noting that if $\eta$ and $\zeta$ are hypercube embeddings then each element of their canonical partitions is also a hypercube embedding.
\end{proof}

\section{Conclusion}\label{sec:conclusion}

Weighted graphs are capable of representing a much richer variety of structures than unweighted graphs. Yet the added complexity of weighted graphs has made it difficult to develop a unified, general theory for understanding isometric embeddings of weighted graphs into various destination graphs of interest, such as unweighted Hamming graphs. Here, we have expanded our understanding of isometric embeddings by relating the pseudofactorization of a minimal weighted graph to its Hamming embeddings, proving that every Hamming embedding is simply a concatenation (or equivalent to a concatenation) of Hamming embeddings for each pseudofactor. Although a polynomial-time algorithm for deciding hypercube embeddability is unlikely to exist \cite{Chvatal,karzanov,AvisDeza}, this eases the task of finding Hamming embeddings, and of proving their non-existence, in cases where the graph has multiple non-trivial pseudofactors. Future work may further characterize the classes of graphs for which we can decide Hamming or hypercube embeddability in polynomial time.

We hope this work also spurs investigation into isometric embeddings of weighted graphs into arbitrary unweighted Cartesian graph products, as opposed to Hamming graphs specifically. In the context of Hamming embeddable graphs, our results imply a hierarchical decomposition of a minimal weighted graph into a Cartesian graph product, whereby a graph may be decomposed first into a Cartesian product of weighted pseudofactors and then each pseudofactor may be individually decomposed into a Cartesian product of (unweighted) complete graphs. An intuitive extension beyond Hamming graphs would be to establish a similar hierarchy for all isometric embeddings of minimal weighted graphs. Such a hierarchy would be an appealing and exciting result, but remains to be proven.

\section*{Acknowledgments}

M.B. and J.B. were supported by the Office of Naval Research (N00014-21-1-4013), the Army Research Office (ICB Subaward KK1954), and the National Science Foundation (CBET-1729397 and OAC-1940231). M.B., J.B., and K.S. were supported by the National Science Foundation (CCF-1956054). Additional funding to J.B. was provided through a National Science Foundation Graduate Research Fellowship (grant no. 1122374). A.C. was supported by a Natural Sciences and Engineering Research Council of Canada Discovery Grant. V.V.W. was supported by the National Science Foundation (CAREER Award 1651838 and grants CCF-2129139 and CCF-1909429), the Binational Science Foundation (BSF:2012338), a Google Research Fellowship, and a Sloan Research Fellowship.

\section*{Declarations of Interest}
The authors declare that they have no competing interests.




\bibliographystyle{plain} 
\bibliography{main.bib}


%% \begin{thebibliography}{00}


%% \bibitem{}

%% \end{thebibliography}
\end{document}